\documentclass[runningheads,a4paper]{llncs}

\usepackage{amssymb}
\usepackage{graphicx}
\usepackage{color}
\usepackage{latexsym}
\usepackage{epsfig}
\usepackage{bm}
\usepackage{amsmath}
\usepackage{amssymb}
\usepackage{amsfonts}
\usepackage{subfigure,wrapfig}
\usepackage{enumerate}

\newcommand{\xty}[2]{\ensuremath{#1}-to-\ensuremath{#2}}
\newcommand{\cvec}{\bm{c}}
\newcommand{\dvec}{\bm{d}}
\newcommand{\fvec}{{\bm{f}}}
\newcommand{\cw}{c.w.\ }
\newcommand{\ccw}{c.c.w.\ }
\newcommand{\wrt}{w.r.t.\ }
\newcommand{\rev}[1]{\ensuremath{\mbox{\it rev}\,(#1)}}
\newcommand{\head}{\ensuremath{\mathit head}}
\newcommand{\tail}{\ensuremath{\mathit tail}}

\newtheorem{invariant}{Invariant}
\newtheorem{observation}{Observation}

\usepackage{url}
\urldef{\mailsa}\path|{glencora, harutyan}@eecs.orst.edu|    
\newcommand{\keywords}[1]{\par\addvspace\baselineskip
\noindent\keywordname\enspace\ignorespaces#1}

\begin{document}

\mainmatter 

\title{Boundary-to-boundary flows in planar graphs}

\author{Glencora Borradaile\inst{1} \and Anna Harutyunyan\inst{2,}\thanks{Work done while at Oregon State University.}}
\institute{Oregon State University \and Vrije Universiteit
  Brussel}


\maketitle

\begin{abstract}
  We give an iterative algorithm for finding the maximum flow between
  a set of sources and sinks that lie on the boundary of a planar
  graph.  Our algorithm uses only $O(n)$ queries to simple data
  structures, achieving an $O(n \log n)$ running time that we expect to be
  practical given the use of simple primitives.  The only existing algorithm
  for this problem uses divide and conquer and, in order to achieve an
  $O(n \log n)$ running time, requires the use of the (complicated)
  linear-time shortest-paths algorithm for planar graphs.

  \keywords{maximum flow, multiple terminal, planar graphs}
\end{abstract}

\section{Introduction}
\setcounter{footnote}{0}
The problem of finding maximum flow in planar graphs has a long history, starting with the
work of Ford and Fulkerson~\cite{FF56} in which the Max-flow, Min-cut
Theorem was proved and the augmenting-paths algorithm was introduced.
Since then, algorithms for maximum flow in planar graphs have fallen
into one of three paradigms: augmenting paths, divide and conquer
using small balanced planar separators, or via shortest paths in the dual.  We
note a subset of these results that are relevant to this paper.
Borradaile and Klein gave an augmenting-paths algorithm for maximum
$st$-flow in directed planar graphs that uses dynamic trees to achieve
an $O(n \log n)$ running time~\cite{BK09}.  For the special case when
$s$ and $t$ are on the same face, an augmenting-paths algorithm can be
simulated via Dijkstra's algorithm or, equivalently, determined from
shortest-path distances in the dual graph~\cite{Hassin81} (details in
Section~\ref{sec:max-flow-shortest}).  Borradaile et~al.\ gave a
rather complicated $O(n \log^3 n)$-time divide-and-conquer algorithm
for when there are multiple sources and sinks (not necessarily on a
common face)~\cite{BKMNW11}. For the special case when these sources
and sinks are all on a common face\footnote{Note that there is no planarity-maintaining
  reduction from this case to the single-source, single-sink
  case.} (such as the boundary of the
embedded graph), Miller and Naor gave a simpler divide-and-conquer
algorithm~\cite{MN95}.

In this work we give an iterative algorithm for this last {\em boundary-to-boundary}
case.  While our algorithm does not improve on the asymptotic running
time of Miller and Naor's work, in order for Miller and Naor's
algorithm to be implemented in $O(n \log n)$ time, one requires repeated applications of 
the linear-time shortest-paths algorithm of Henzinger et
al.~\cite{HKRS97}.  This shortest-paths algorithm is arguably
impractical: it is also a divide-and-conquer algorithm using small
planar separators, involves `large constants' and, to our knowledge,
has not been implemented.  Our algorithm, on the other hand, requires
just $O(n)$  (with a small constant) queries to simple data structures:
namely a priority queue and a linked list~\cite{DS87}.

Our algorithm is an augmenting-paths algorithm that iterates
over the source-sink pairs.  We simulate finding the flow between a
given source and sink using Hassin's method -- via Dijkstra's algorithm in the dual graph.  In order to prevent searching the same region of the
graph multiple times, we search the graph in a biased way~\cite{GH05},
such that we need only reuse the boundary of the searched region for
augmenting further source-sink pairs.  In order to reuse these
boundaries efficiently, we use a simple generalization of priority
queues in which queues are merged whose relative priorities differ by a
constant or {\em offset}.  These {\em offset queues} are implemented using edge weights to encode the offset in a tree implementation of the heap;  doing so does not affect the asymptotic running time of the basic priority queue operations.  Details are given in Appendix~\ref{app:pq}.

We believe that the methods used in this paper may be applicable to
other planar flow problems.  For example, in a companion paper~\cite{BH13}, we
argue that the augmenting-paths algorithm of Borradaile and Klein for
maximum $st$-flow in directed planar graphs can also be simulated by
Dijkstra in the dual graph; the details of the implementation in this
paper may lead to an $O(n \log n)$ algorithm for maximum $st$-flow in
directed planar graphs that does not require the more cumbersome
dynamic-trees data structure.

\subsection{Definitions}

We give a brief outline of definitions where we may stray from
convention. For more complete and formal definitions, please refer to
Borradaile's dissertation~\cite{Borradaile-thesis}. We extend any function or property on elements to sets of elements in
the natural way.

Our algorithms are for directed graphs, but we consider the underlying
undirected graph where each edge has two oppositely directed darts.  Darts are oriented from {\em tail} to {\em head}.
Capacities, $\cvec$, on the darts are positive and asymmetric, reflecting  the
original directed problem.  Paths and cycles are sequences of darts
and so are naturally directed; a path or a cycle may visit the same vertex multiple times; those that
do not are {\em simple}; a path may be trivial, in which case it
is a vertex. $X[a,b]$ denotes the $a$-to-$b$ subpath of $X$ where $X$
is a path, cycle or tree;  $\circ$ denotes the concatenation of paths (which may result in a cycle).

A flow $\fvec$ is an assignment of  real numbers to darts that is antisymmetric (for a dart and its reverse),
respects capacities and is balanced at all non-terminal (non-source,
sink) vertices.  The value $|\fvec|$ of a flow is the net flow
entering the sinks.  A flow is a circulation if there are no
terminals.  The residual capacities $\cvec_f$ of capacities $\cvec$
\wrt flow $\fvec$ are given by:
\begin{equation}
  \label{eq:res}
  \cvec_\fvec[d] = \cvec[d]-\fvec[d],\ \forall \text{darts } d
\end{equation}
A path or cycle $X$ is residual if the residual capacity of every dart in $X$ is strictly
positive.  A dart is saturated if its residual capacity is zero.
Residuality is \wrt capacities (such as $\cvec$ or $\cvec_\fvec$).

An $xy$-cut in $G$ is a set of darts $C$, the removal of which leaves
no \xty{x}{y} paths. The value of a cut is the total capacity of its darts. The
value of the {\em minimum} $xy$-cut equals to that of the maximum $xy$-flow~\cite{FF56}.

We use the usual definitions for planar graphs and their duals.  We
denote any path, cycle, vertex, face, dart in the dual graph with a
$*$-superscript.  If $d$ is a dart in $G$, then $d^*$ is the
corresponding dual dart; if $v$ is a vertex and $f$ is a face in $G$,
$v^*$ is a {\em face} and $f^*$ is a {\em vertex} in $G^*$.  The boundary of the
graph is denoted $\partial G$ and is taken to be clockwise.  We refer
to simple cycles as being clockwise (c.w.) or counterclockwise
(c.c.w.); \cw and \ccw depend on the choice of infinite face, $f_{\infty}$, which,
throughout this paper, we will take to be the face common to all the
sources and sinks.

For two non-crossing $x$-to-$y$ paths $P$ and $Q$, we say $P$ is left
of $Q$ if $P \circ \rev{Q}$ is c.w. A path is leftmost if there
are no paths left of it. For an $x$-to-$y$ path $P$ that
starts and ends on $\partial G$, we say a face, edge, path, etc.\ $X$
is (strictly) left of $P$ if $X$ is (strictly) contained by the \cw
cycle $\partial G[x,y] \circ \rev{P}$.  We say that a planar flow
$\fvec$ is {\em leftmost} if every \cw cycle is non-residual
\wrt $\cvec_\fvec$.  We say that capacities are \cw acyclic if every
\cw cycle is non-residual \wrt the capacities.

\section{Leftmost maximum flows and shortest paths}
\label{sec:max-flow-shortest}
Khuller, Naor and Klein~\cite{KNK93} showed that a flow that is
derived from shortest-path distances in the dual is \cw acyclic.  Formally:
\begin{theorem}[Clockwise acyclic flows] \label{thm:cwsat}
  Let $\dvec$ be the shortest-path distances in $G^*$ from $f_\infty^*$
  interpreting capacities as lengths.  Then every \cw cycle is non-residual \wrt the flow
  \begin{equation}
    \fvec[d] =
    \dvec[\head(d^*)]-\dvec[\tail(d^*)] \ \forall \text{ darts }
    d \label{eq:st-flow}
  \end{equation}
  where $\head(d^*)$ and $\tail(d^*)$ are the head and tail vertices of $d^*$ in $G^*$.
\end{theorem}

Earlier, Hassin had used this idea to find a maximum $st$-flow in an
$st$-planar graph~\cite{Hassin81}.  We can view his algorithm by turning it into a
circulation problem: introduce a new infinite-capacity arc $ts$
embedded so that every $s$-to-$t$ residual path forms a \cw
cycle with $ts$ and then saturate the \cw cycles.  We describe
an equivalent formulation which we use in this paper.  Split the dual
vertex $f_\infty^*$ into two vertices $a_\infty^*$ and $b_\infty^*$
such that all the darts in $\partial G[s,t]^*$ are incident to
$a_\infty^*$ and all the darts in $\partial G[t,s]^*$ are incident to
$b_\infty^*$; denote the resulting graph $G_{st}^*$.  Let $\dvec[x^*]$
be the shortest-path distance from $a_\infty^*$ to $x^*$ in $G^*_{st}$,
viewing capacities as lengths.  Then the flow assignment $\fvec_{st}$
for $G$ given as in Equation~(\ref{eq:st-flow}) is a maximum
$st$-flow.  It follows directly from Theorem~\ref{thm:cwsat} that
$\fvec_{st}$ is the leftmost maximum $st$-flow.

Since simple cuts in the primal map to simple cycles in the
dual (and vice versa)~\cite{Whitney1933}, the darts of an $st$-cut $C$ form an
$a_\infty^*$-to-$b_\infty^*$ path $C^*$ in $G_{st}^*$.  If $C$ is a
minimum cut, $C^*$ is a shortest path.

\begin{observation}\label{obs:leftmost-acyclic}
  A leftmost flow \wrt \cw acyclic residual capacities is acyclic.~\cite{BK09}
\end{observation}
Because of this acyclicity, one can easily show:
\begin{observation}\label{obs:path-decomp} 
  Let $\cvec$ be \cw acyclic capacities and let $\fvec$ the leftmost,
  max $st$-flow for $s$ and $t$ on $f_\infty$. Then there is a
  decomposition of $\fvec$ into unique, non-crossing $s$-to-$t$ paths
  $P_1, P_2,\ldots, P_\ell$ where $P_i$ carries $f_i > 0$ units of flow
  and $P_i$ is left of $P_j$ $\forall i < j$.  Further, an
  augmenting-paths algorithm that always saturates the leftmost path
  first saturates the paths $P_1, \ldots, P_\ell$ in order.
\end{observation}
Our algorithm requires \cw acyclic capacities; the analysis will use
this fact indirectly by invoking Observation~\ref{obs:path-decomp}.
We will achieve this property in a preprocessing step and maintain
this as an invariant throughout the algorithm.  It follows from
Equation~(\ref{eq:st-flow}) and Observation~\ref{obs:path-decomp}
that, for every primal face $x$ (dual vertex $x^*$):
\begin{equation}
  \label{eq:distances}
  \dvec[x^*] = \left\{
    \begin{array}[c]{ll}
      \sum_{j = 1}^i f_j & \text{if $x$ is right of $P_i$ and left of $P_{i+1}$} \\
      \sum_{j = 1}^\ell f_j = |\fvec| & \text{if $x$ is right of $P_\ell$}
    \end{array}\right.
\end{equation}

\subsection{$st$-planar flow via biased search} \label{sec:biased}

We describe how to find an $st$-planar flow via biased search (in the
dual) that does not necessarily search the entire graph, assuming that
the initial capacities are \cw acyclic.  We assume that there are no
degree-2 vertices in the primal; any such vertex could be removed by
merging the adjacent darts (in each direction) and keeping the minimum
of the capacities.  Parallel darts (not antiparallel) can be merged by taking the sum of
their capacities.  We additionally assume that the finite faces of the
primal are triangulated (which can be achieved by the addition of
0-capacity edges).

\begin{wrapfigure}{r}{0.5\textwidth}
  \vspace{-30pt}
  \centering
  \includegraphics[scale = 0.7]{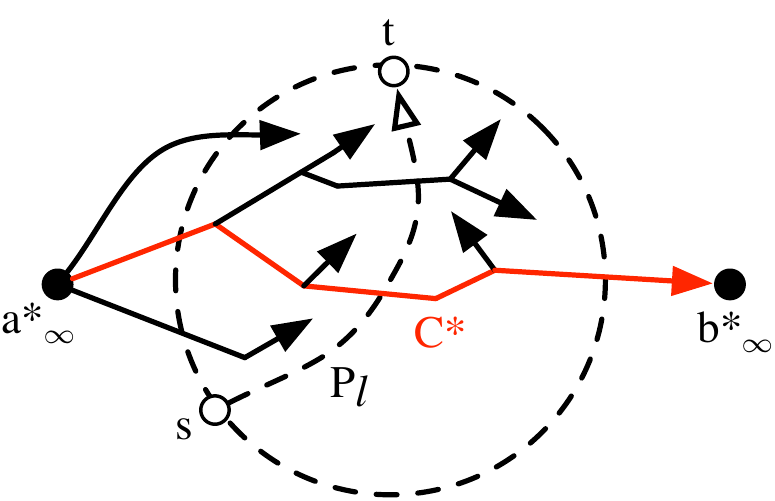}
  \caption{$\partial G$ is the dashed circle and the dashed $s$-to-$t$
    path is $P_\ell$.  In $G_{st}^*$, $a_\infty^*$ is incident to the
    duals of all the arcs on the path of the circle \cw from $s$
    to $t$ and $b_\infty^*$ is incident to the duals of all the arcs
    on the path of the circle \cw from $t$ to $s$.  The solid
    tree is the search tree used in the biased search algorithm with
    the $a_\infty^*$-to-$b_\infty^*$ path representing the leftmost
    cut $C^*$.}
  \label{fig:witness}
  \vspace{-30pt}
\end{wrapfigure}

We implicitly and iteratively build a decomposition as given in
Observation~\ref{obs:path-decomp} using Dijkstra's algorithm in the
dual.  Initially $P_1 = \partial G[s,t]$. In phase $i$, we have
already found path $P_i$; we maintain that, at the start of phase $i$,
the faces adjacent to and right of $P_i$ are in the queue $Q_i$. (Keep
in mind that faces are vertices in the dual, and we are really just
finding shortest paths in the dual graph, applying the standard rules
for Dijkstra's algorithm.)  The priority of face $x$ is the capacity
of the minimum-capacity dart bounding $x$ in $P_i$.  Say the minimum
priority in the queue is $q$; to find $P_{i+1}$ we pop faces off the
queue with priority $q$ until the minimum priority in the queue is
$>q$.  Now we have popped off all the faces between $P_i$ and
$P_{i+1}$ (by Equation~(\ref{eq:distances})) and $Q_{i+1}$ contains
all the faces to the right of and adjacent to $P_{i+1}$.

So far, we have just described Hassin's algorithm, but have made
explicit the augmenting paths that are implicit in his algorithm.  We
have also identified {\em phases}.  In each phase, all the faces of a
given distance label are explored via 0-length darts (in the
dual).

We modify the algorithm so that we do not explore the entire graph.
Note that all the faces to the
right of $P_\ell$ (the last augmenting path), by Equation~(\ref{eq:distances}) have distance label
$|\fvec|$.  Rather than label all these faces, after
getting to the start of phase $\ell$, we wish to find the {\em leftmost cut}.
Let $C^*$ be the leftmost, shortest
$a_\infty^*$-to-$b_\infty^*$ path in $G_{st}^*$; $C$ is the leftmost cut.  The part of
$C^*$ that is strictly to the right of $P_\ell$ consists of 0-length
darts, since the sum of the capacities of the darts in $C^*$ that are
in $P_1, \ldots, P_\ell$ is $|\fvec|$ by Equation~(\ref{eq:distances}).  In addition to identifying the
leftmost cut, we wish to not explore any part of the graph strictly
right of $P_\ell$ and $C^*$.  (See Figure~\ref{fig:witness}.)

We find the leftmost cut by at each phase additionally maintaining an ordering $A_i$ of the faces in
$Q_i$ that reflects their order along $P_i$ from $t$ to $s$.  We
maintain and query the ordering using the order maintenance data
structure {\sc DSOrder} due to Dietz and Sleator~\cite{DS87} which is
a circularly linked list with order information determined using 2's
complement arithmetic.  (See Appendix~\ref{app:order} for
details. Each of the operations takes either $O(1)$ or $O(\log n)$
time per visited face.)  During a phase, we:

(1) Start with faces that are closest to $t$ in the ordering.

(2) Explore along 0-length darts in the dual in a depth-first {\em
    leftmost} fashion; this can be done by following the combinatorial
  embedding of the darts around a vertex in a \cw order, using the
  parent dart in the search tree implicit to Dijkstra's
  algorithm~\cite{BK09}.

(3) If we reach $b_\infty^*$ during this search, we immediately
  stop the algorithm.  (More details of this are given below.)

(4) At the end of this 0-length exploration, we remove from the
  queue and order any faces that we have reached in this exploration.
  Suppose $T^*$ is the dual search tree we have explored that contains the shortest
  paths found by Dijkstra's algorithm, rooted at a face
  adjacent to $P_i$.  We add the never-visited faces adjacent to $T^*$ in
  their \cw order around $T^*$ (according to their shortest adjacency
  to $T^*$).  This ordering is easily visualized by contracting the
  edges of $T^*$ and considering the \cw ordering of the darts around
  the new (dual) vertex.

At the start of each phase, the queue and the order contain the same
set of elements. The leftmost-bias to the search additionally guarantees that the final dual
search tree $T^*$ contains leftmost shortest paths.  This can be
easily shown via induction.  Since we stop as soon as we reach
$b_\infty^*$ and we search in a leftmost fashion, $T^*$ does not
contain any darts {\em strictly} right of both the last flow path
$P_\ell$ found and $T^*[a_\infty^*,b_\infty^*]$.  In this way, we also
guarantee: 
\begin{observation} \label{obs:end} At the end of this biased search,
  the queue and order contain the faces adjacent to and right of
  $P_\ell$.  
\end{observation}
In our multi-source, multi-sink algorithm, we will reuse this queue
and order.  To do so, we need to know the residual capacities
of the darts in $P_\ell$.  If a face $f$ in the queue has exactly one
bounding arc in $P_\ell$, then the priority of $f$ reflects exactly
the residual capacity of that dart.  If $f$ has two bounding darts
$d_1$ and $d_2$ in $P_\ell$ (i.e., the head of $d_1^*$ and $d_2^*$ in
$G^*$ is $f^*$), then, {\em to the right of $P_\ell$}, we can only
push the minimum of these darts' residual capacities along this section of
$P_\ell$.  (Put another way, if we remove everything strictly to the
left of $P_\ell$, $d_1$ and $d_2$ would be incident to a degree 2
vertex, which we would remove according to the rule at the start of
this section.)  We get:
\begin{observation} \label{obs:res-cap}
  The priority of a face $f$ in the queue reflects the residual
  capacity of the dart(s) bounding the face in $P_\ell$; the residual
  capacity is the priority less $|\fvec|$.
\end{observation}
Subtracting $|\fvec|$ from the priorities in the ending queue can be
done in $O(1)$ time using offset queues (Appendix~\ref{app:pq}).
Finally, the {\sc DSOrder} data structure does not allow us to pull
the {\em first} element of the order (having minimum priority in the
queue) but does allow us to sort a subset of items.  In doing so, we
spend $O(\log n)$-amortized time per element.  We do not wish to
repeat this work.  If we reach $b_\infty^*$ in the middle of a phase
and have a subset of items $X$ that we have sorted using {\sc
  DSOrder}, we break the ties in the priorities of these items in the
priority queue.  When we return to use this queue/order, we will not
need to resort these items.

\section{Algorithm}
\label{sec:algorithm}

For simplicity of presentation we will assume that the terminals are
alternating sources and sinks along $\partial G$.  This can be
attained by taking a consecutive group of sources $S$, introducing a
new source and connecting the new source to every source in $S$ with
an infinite capacity arc.  We number the sources and sinks according
to their \cw ordering on $\partial G$, $s_1, t_1, s_2, t_2, \ldots,
s_m, t_m$, starting with an arbitrary source.  We return the
difference between the original capacities and final residual
capacities, which, by Equation~(\ref{eq:res}), is the corresponding
flow.

\begin{tabbing}\label{tab:algo-super-high-level}
  {\sc AbstractFlow} ($G$, $\{s_1, t_1, s_2, t_2, \ldots, s_m, t_m\}$, $\cvec$)\\
  \qquad Saturate all $s_j$-to-$t_i$ residual paths $\forall i < j$ and all \cw cycles.\\
  \qquad Let $\cvec_0$ be the resulting residual capacities.\\
  \qquad For $j =1, 2, \ldots, m$:\\
  \qquad \qquad for $i = j, j-1, \ldots, 1$:\\
  \qquad \qquad \qquad \= let $\cvec_{ij}'$ be the current residual capacities.\\
  \> Find the leftmost \xty{s_i}{t_j} flow $\fvec_{ij}$ \wrt $\cvec_{ij}'$.\\
  \> Let $\cvec_{ij}$ be the residual capacities of $\cvec_{ij}'$ \wrt $\fvec_{ij}$.\\
  \qquad Return $\cvec[d]-\cvec_{mm}[d]$ for all darts $d$.
\end{tabbing}

\begin{figure}
  \centering
  (a) \includegraphics[scale = 0.8]{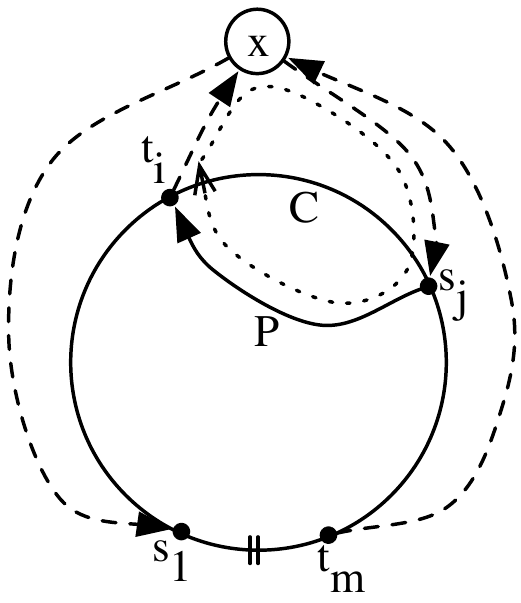}
  (b) \includegraphics[scale=0.8]{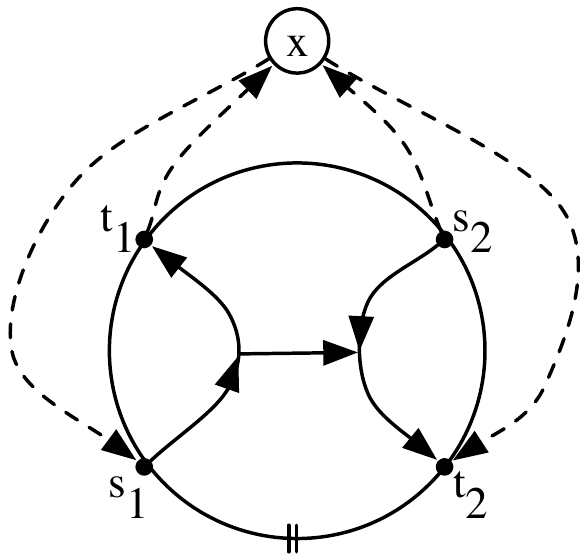}
  \caption{(a) Illustrating the first step of the {\sc AbstractFlow}.
    (b) A simple example illustrating why this first step cannot be
    repeated to find the overall maximum flow.  The equivalent step
    would saturate all \ccw cycles.  If the solid edges have equal
    capacity, this would saturate the $s_1$-to-$t_2$ path, since the
    method for saturating all \ccw cycles (like for \cw cycles)
    saturates all largest such cycles.  However, doing so would create
    a residual path from $s_2$ to $t_1$.}
  \label{fig:preprocess}
\end{figure}

The first step can be done with one shortest-path computation in the
dual as follows (in $O(n \log n)$ time using Dijkstra's algorithm, for
example); refer to Figure~\ref{fig:preprocess}(a).  Embed a vertex $x$
in $f_\infty$.  Connect $x$ to every source and every sink with
infinite-capacity arcs.  Embed these arcs so that $s_1$, $t_m$ and $x$ are
on the infinite face.  Let $\fvec$ be the circulation that saturates
all the \cw residual cycles in this graph (Theorem~\ref{thm:cwsat}).
Let $\cvec_0$ be the residual capacities of the darts in $G$ \wrt
$\fvec$.  Consider any simple path $P$ from $s_j$ to $t_i$ in $G$.
For $j > i$, $P \circ t_ix \circ xs_j$ is a \cw cycle $C$.  Therefore
$C$ must be non-residual \wrt $\cvec_0$ and, since the arcs $t_ix$ and
$xs_j$ have infinite capacity, $P$ must be non-residual \wrt
$\cvec_0$.

Note that while the iterative part of the algorithm saturates all
$s_i$-to-$t_j$ paths $\forall i < j$, we cannot achieve this with a
symmetric application of the first step.  The simple example in
Figure~\ref{fig:preprocess}(b) illustrates why.

In the remainder of the paper we will give an efficient implementation
of the double loop of {\sc AbstractFlow}.  We first show that
the abstract algorithm guarantees several useful invariants that limit
the region of the graph that is involved in each iteration.  These
invariants allow us to explore the graph in such a way that no region
is explored multiple times.  Correctness of {\sc AbstractFlow}
will also follow from these invariants.  
By iteration $i,j$, we will mean iteration $i$ of the inner loop and
iteration $j$ of the outer loop.  

\subsection{Invariants}
\label{sec:analysis}

Since only leftmost flows are augmented we get (by definition and induction):
\begin{invariant}
  \label{inv:no-cw-cycles}
  There are no clockwise residual cycles in $G$ \wrt $\cvec_{ij}, \forall i \le j$.
\end{invariant}

Since the sink is in common to all the iterations of the inner loop,
for a given iteration of the outer loop, we get:
\begin{invariant} \label{inv:inner-loop}
  There are no residual $s_j$-to-$t_k$ paths \wrt $\cvec_{i,k}$ for $j > i$.
\end{invariant}
More formally, this follows from the Sinks Lemma~\cite{BKMNW11}.
The following invariant shows that we do not undo the progress made by the first step of {\sc AbstractFlow}.

\begin{invariant}\label{inv:no-back-flow}
  There are no \xty{s_i}{t_j} residual paths s.t. $i > j$ \wrt
  $\cvec_0$ or $\cvec_{k{\ell}}$, $\forall k < \ell$.
\end{invariant}

\begin{proof}
  We prove this invariant by induction. It holds \wrt $\cvec_0$ as
  argued in Section~\ref{sec:max-flow-shortest}.  For a contradiction, let
  $\cvec_{k\ell}$ be the first residual capacities that introduce an
  \xty{s_i}{t_j} residual path $R$ ($i< j$).  Then there must be an
  \xty{s_k}{t_{\ell}} path $A$ that is augmented in iteration $k,\ell$
  and that uses a dart $d$ in $\rev{R}$.

  Let $x$ and $y$ be the last and first, resp., vertices of $R$ that
  are in $A$.  $A$, $R[s_i,y]$ and $R[x,t_j]$ are residual \wrt
  $\cvec_{k\ell}'$ (the residual capacities at the start of iteration $k,\ell$).  It follows that $k \le j$ and $\ell > i$, for
  otherwise we contradict the inductive hypothesis.  However,
  iteration $k,\ell$ comes after $i,\ell$ in {\sc
    AbstractFlow}.  Invariant~\ref{inv:inner-loop} tells us
  that there cannot be an $s_i$-to-$t_\ell$ path that is residual \wrt
  $\cvec_{k\ell}'$, contradicting the existence of $R[s_i,y] \circ
  A[y,t_\ell]$.  \qed
\end{proof}

The optimality of the flow found by {\sc AbstractFlow} follows from the last invariant (along with Invariants~\ref{inv:inner-loop} and~\ref{inv:no-back-flow}):
\begin{invariant} \label{inv:overall-progress}
  There are no $s_i$-to-$t_j$ residual paths \wrt $\cvec_{\ell k}$ for any $\ell$ and any $k > j$.
\end{invariant}

\begin{proof}
  We prove this invariant by induction. It holds \wrt $\cvec_{1,j+1}'$
  by Invariant~\ref{inv:inner-loop}. For a contradiction, let
  $\cvec_{\ell k}$ be the first residual capacities that introduce an
  \xty{s_i}{t_j} residual path $R$.  W.l.o.g.\ assume that $i \le j$
  as the case $i > j$ is handled by Invariant~\ref{inv:no-back-flow}.
  Then there must be an \xty{s_\ell}{t_k} path $A$ that is augmented
  in iteration $\ell,k$ and that uses a dart $d$ in $\rev{R}$.

  Let $x$ and $y$ be the first and last, resp., vertices of $R$ that
  $A$ shares.  Since $A$ and $R[y, t_j]$ are residual, $\ell \le j$ by
  Invariant~\ref{inv:no-back-flow}.  However, by
  Invariant~\ref{inv:inner-loop}, there are no $s_\ell$-to-$t_j$ paths
  that are residual \wrt $\cvec_{1j}$, so $\ell > j$, a contradiction. \qed
\end{proof}

\subsection{Unusability Structures}

We will illustrate our implementation of {\sc AbstractFlow} with a
recursive algorithm.  To that end, we show that the cut and the flow
found in iteration $i,j$ separates the graph into two pieces that act
independently for the remainder of the algorithm.  Let $P$ be the
rightmost path in the path decomposition of $\fvec_{ij}$ given in
Observation~\ref{obs:path-decomp} (that has non-zero flow).  The
following lemma allows us to delete everything strictly to the left of
$P$ at the end of iteration $i,j$ for future iterations without affecting optimality.

\begin{lemma} \label{lem:path-bdy}
  There are no paths from $s_k$ to $P$ that are residual \wrt $\cvec_{ij}$ for $k > i$.
\end{lemma}

\begin{proof}
  First we make an observation.  Inner iterations $j, j-1, \ldots, i$
  are equivalent to adding a new source $s$, connecting $s$ to $s_j,
  s_{j-1}, \ldots, s_i$ by high-capacity arcs and saturating the
  leftmost max $st_k$-flow\footnote{Note that in the implementation,
    we do not merge the sources in this way as doing so does not allow
    us to reuse the work done in previous iterations.}.  By
  Observation~\ref{obs:path-decomp}, this is done by saturating a set
  of non-crossing $s$-to-$t_k$ paths ${\cal P} = P_1, P_2, \ldots$
  ordered from left to right.  In {\sc AbstractFlow}, iteration
  $\ell,k$ will saturate a contiguous subset ${\cal P}_\ell$ of ${\cal P}$ for $i \le
  \ell \le j$.  By
  saturating these paths in order, we first cut $s_j$ from $t_k$ by saturating ${\cal P}_j$, then
  cut $s_{j+1}$ from $t_k$ and so on.

  For $i < k \le j$, the lemma follows from the fact that iteration
  $k,j$ precedes $i,j$: a path $Q$, from $s_k$-to-$P$ concatenated
  with the suffix of $P$, would be saturated before $P$.  For $k > j$,
  $Q$ would be residual \wrt capacities $\cvec_{ij}'$ since
  $\fvec_{ij}$ does not change the capacities of darts strictly to the
  right of $P$; $Q$ violates Invariant~\ref{inv:no-back-flow}.  \qed
\end{proof}

Let $C$ be the leftmost minimum $s_it_j$-cut.  The next lemma shows
that we can delete the darts in $C$ (among others on the $t_j$ side of
the cut) without affecting optimality.  In the biased-search
algorithm (Section~\ref{sec:biased}), the darts satisfying
Lemma~\ref{lem:unusable-cuts} are exactly those that are searched to the
right of the last flow path ($T^*$) in finding the leftmost cut ($C$).

\begin{lemma} \label{lem:unusable-cuts} Let $W^*$ be any
  from-$a_\infty^*$ $|\fvec_{ij}|$-length path in $G_{s_it_j}^*$ that
  is left of $C^*$.  Then no $s$-to-$t$ path that is residual w.r.t.\
  $\cvec_{ij}$ uses a dart in $W$.
\end{lemma}

\begin{proof}
  For a contradiction, suppose there is a $s_k$-to-$t_\ell$ path $R$
  that is residual \wrt $\cvec_{ij}$ that uses a dart of $W$.  Since,
  by Invariant~\ref{inv:no-back-flow}, $\ell \ge k$, $s_k$ must be on
  the $t_j$ side of $C$ for otherwise, $R$ would have to cross back
  and forth across $C$, but the darts of $C$ are only residual \wrt
  $\cvec_{ij}$ from the $t_j$ side to the $s_i$ side.

  We have just finished iteration $i,j$, $k > j$, and so, by
  Invariant~\ref{inv:no-back-flow}, there is an $s_kt_j$-cut
  $K$. Take $K$ to be the {\em rightmost} of these cuts (defined
  analogously to leftmost).  In $G_{s_it_j}^*$, $K^*$ is a \ccw cycle
  through $b_\infty^*$; $K^*$ is 0-length (or, equivalently, 
  composed entirely of darts that are non-residual \wrt $\cvec_{ij}$).

  $K^*$ must be left of $C^*$, for otherwise, the leftmost-ness of
  $C^*$ and the rightmost-ness of $K^*$ would be violated.  If $R$
  uses a dart $d$ of $W$, then $d$ must be on the $s_k$ side of $K$.
  Then, in the dual, $W^*$ must intersect $K^*$ at a dual vertex
  $x^*$.  But then $W^*[a_\infty^*,x^*] \circ K^*[x^*,b_\infty^*]$ is
  a $a_\infty^*$-to-$b_\infty^*$ path of length at most that of $W^*$;
  $W^*[a_\infty^*,x^*] \circ K^*[x^*,b_\infty^*]$ is left of $C^*$,
  contradicting that $C$ is a leftmost cut.\qed

 \end{proof}

Lemmas~\ref{lem:path-bdy} and~\ref{lem:unusable-cuts} allow us to
implement {\sc AbstractFlow} recursively.  That is, {\sc
  AbstractRecursiveFlow}, below, finds the same (non-zero) flows
$\fvec_{ij}$ in the same order as {\sc AbstractFlow}.  The recursive
algorithm has a slightly different input, as there may be several
consecutive sources for the recursive calls.  We illustrate the
algorithm without explicitly returning the flow.  It is trivial to
determine the flow from the residual capacities found throughout the
algorithm.
\newpage

{\footnotesize \begin{tabbing}
  {\sc AbstractRecursiveFlow}($G$, $\{s_1, t_1, \ldots, s_m, t_m\}$, $\cvec$)\\
  \qquad \= Saturate all $s_j$-to-$t_i$ residual paths $\forall i < j$ and all \cw cycles.\\
  \> Let $\cvec_0$ be the resulting residual capacities. \\
  \> {\sc AbstractRecursiveFlowHelper} ($G$, $\{\}$, $\{s_1, t_1, \ldots, s_m, t_m\}$, $\cvec_0$)
 \end{tabbing}

\begin{tabbing}
  {\sc AbstractRecursiveFlowHelper}($G$, $\{s_1, s_2, \ldots, s_{\ell-1}$\}, $\{s_\ell, t_\ell, s_{\ell+1}, t_{\ell+1}, \ldots, s_m, t_m\}$, $\cvec$)\\
  \qquad\= Find the leftmost $s_\ell$-to-$t_\ell$ flow $\fvec$ \wrt $\cvec$. \\
  \> Let $\cvec'$ be the residual capacities of $\cvec$ \wrt $\fvec$. \\
  \> Let $P$ be the rightmost path in the path-decomposition of $\fvec$ and let $C$ be the leftmost cut. \\
  \> Let $G_1$ and $G_2$ be the components resulting from deleting all the darts \\
  \> \qquad \= \qquad strictly to the left of $P$ and the darts of $C$ from $G$. \\
  \> If $t_\ell \in G_2$: \\
  \> \> Let $k$ be the greatest index s.t.\ $t_k \in G_2$. \\
  \> \> {\sc AbstractRecursiveFlowHelper}($G_2$, $\{\}$, $\{s_{\ell+1}, t_{\ell+1}, \ldots, s_k, t_k\}$, $\cvec'$) \\
  \> \> Let $h$ be the smallest index $> \ell$ s.t.\ $t_h \in G_1$. \\
  \> \> Extend $Q$ and $A$ to contain all the faces in $G^*_{1,s_{\ell}t_h}$ that are incident to $a_\infty^*$ \\
  \>\> {\sc AbstractRecursiveFlowHelper}$(G_1, \{s_1, s_2, \ldots,
  s_{\ell}\}$, $\{s_h, t_h,\ldots, s_m,
  t_m\},\cvec')$\\
  \> Else: \\
  \> \> Let $j$ be the greatest index $<\ell$ s.t.\ $s_j \in G_1$.\\
  \>\> {\sc AbstractRecursiveFlowHelper}($G_1$, $\{s_1, s_2, \ldots,
  s_j\}$, $\{s_{\ell+1}, t_{\ell+1}, \ldots, s_m, t_m\}$, $\cvec'$) \\
\end{tabbing}}

\begin{wrapfigure}{r}{0.55\textwidth}
  \vspace{-30pt}
  \centering
  \includegraphics[scale=0.7]{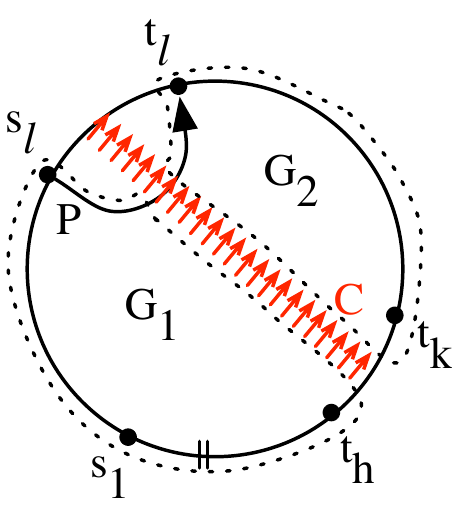}
  \includegraphics[scale=0.7]{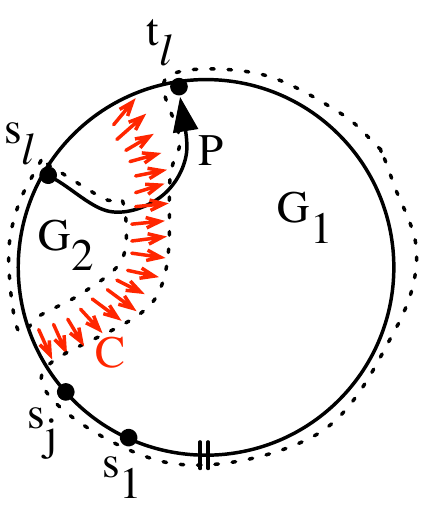}
  \caption{The two cases for subproblems for {\sc
      (Abstract)RecursiveFlow}. If $t_\ell \in G_2$ (left), there are 2
    non-trivial subproblems.}
  \label{fig:subprob}
  \vspace{-30pt}
\end{wrapfigure}

\begin{lemma}
  {\sc AbstractRecursiveFlow} implements {\sc AbstractFlow}.
\end{lemma}

\begin{proof}
  Refer to Figure~\ref{fig:subprob}.  By Lemmas~\ref{lem:path-bdy}
  and~\ref{lem:unusable-cuts}, the deleted edges are {\em safe} to
  remove: solving the problem in the two subproblems will indeed find
  an optimal solution.  The $s_\ell t_\ell$ augmentation performed by
  {\sc AbstractRecursiveFlowHelper} corresponds to an iteration of
  {\sc AbstractFlow}.  If there are residual source-to-$t_\ell$ paths
  remaining after this augmentation, then there would necessarily be
  one such path from $s_1$, and $t_\ell \notin G_2$.  {\sc
    AbstractRecursiveFlowHelper} would continue to push flow from
  earlier sources to $t_\ell$, just as {\sc AbstractFlow}.  Otherwise,
  both {\sc Abstract}- and {\sc AbstractRecursive-Flow} would move onto
  the next sink, in which case $t_\ell \in G_2$.  \qed
\end{proof}

\subsection{Reusing queues for an efficient implementation}

We show how to implement {\sc AbstractRecursiveFlow} using $O(n)$
queries to simple data structures: the priority queue and {\sc
  DSOrder} data structure (which is at heart a linked list).  The
challenge in doing so can be illustrated by a simple example.  Suppose
$s_1$ has a high-capacity path $P$ with many edges ending with a
low-capacity star that connects to each of the sinks.  In each
iteration of the outer loop, we could require augmenting the flow
along this long path.  We overcome this barrier by reusing the work
from earlier iterations in later iterations.

We give an implementation ({\sc RecursiveFlow}) of {\sc
  AbstractRecursiveFlowHelper}.  To implement the first step of {\sc
  AbstractRecursiveFlowHelper}, we use the biased-search algorithm
described in Section~\ref{sec:biased}.  Note that the subproblem
corresponding to terminal sets $\{s_1, s_2, \ldots, s_{\ell-1}\}$, $\{s_\ell, t_\ell,
s_{\ell+1}, t_{\ell+1}, \ldots, s_m, t_m\}$ results from having found
maximum flows from $s_1, s_2, \ldots, s_{\ell-1}$ to $t_\ell$.  We
keep the queue and order at the end of the biased-search algorithm
used to find these flows.  

Formally, we will pass to {\sc RecursiveFlow} a queue and order for
each source $s_i$, $i \le \ell$.  The queue $Q_i$ and order $A_i$
contains all the faces adjacent to and right of $\partial
G[s_i,s_{i+1}]$ for $i < \ell$ and $\partial G[s_i,t_i]$ for $i =
\ell$.  The order reflects the \ccw ordering of the faces along
$\partial G$.  The priority of a face $f$ in $Q_i$ is the current
residual capacity of the primal copy of the dart $f_\infty^*f^*$. Recall from Section~\ref{sec:biased} that the biased-search algorithm guarantees this at the end of the search.

{\footnotesize
\begin{tabbing}
  {\sc RecursiveFlow}$(G, \{s_1, s_2, \ldots, s_{\ell-1}\}$, $\{s_\ell, t_\ell, s_{\ell+1}, t_{\ell+1}, \ldots, s_m, t_m\}$, $\{(Q_1,A_1), \ldots, (Q_\ell,A_\ell)\})$\\
  1\qquad \= Find the leftmost $s_\ell t_\ell$-flow $\fvec$ via biased-search using $Q_\ell, A_\ell$ as the starting queue, order.\\
  2\> Let $P$ be the rightmost path in $\fvec$ and let $T^*$ be the search tree. \\
  3\> Let $Q, A$ be the queue and order at the end of this search. \\
  4\> Subtract $|\fvec|$ from the priorities in $Q$. \\
  5\> Delete everything to the left of $P$ in $G$.  \\
  6\> Delete from $G$ the darts in $T^*$ that are left of $P$ creating components \\
  \> \qquad \= \qquad \qquad $G_1$ (that contains $s_1$) and $G_2$.\\
  7\> If $t_\ell \in G_2$: \\
  8\> \> Initialize the queue $Q_{\ell+1}$ and ordering $A_{\ell+1}$ of the dual vertices \\
  \> \> \qquad \qquad \qquad adjacent to $a_\infty^*$ in $G^*_{s_{\ell+1}t_{\ell+1}}$ \\
  9\> \> Let $k$ be the greatest index s.t.\ $t_k \in G_2$. \\
  10\> \> {\sc RecursiveFlow}$(G_2, \{\}, \{s_{\ell+1}, t_{\ell+1}, \ldots, s_k, t_k\}, \{(Q_{\ell+1},A_{\ell+1})\})$\\
  11\> \> Let $h$ be the smallest index $> \ell$ s.t.\ $t_h \in G_1$. \\
  12 \> \> Extend $Q$ and $A$ to contain all the faces in
  $G^*_{1,s_kt_h}$ that are incident to $a_\infty^*$ \\
  \>\> \qquad \qquad not currently in $Q/A$ with the appropriate priority/order.\\
  13\>\> {\sc RecursiveFlow}$(G_1, \{s_1, s_2, \ldots, s_{\ell}\}$,
  $\{s_h, t_h, \ldots, s_m, t_m\}, \{(Q_1,A_1), \ldots, (Q_{\ell},A_{\ell}), (Q,A)\})$\\
  14 \> Else: \\
  15 \> \> Let $j$ be the greatest index $<\ell$ s.t.\ $s_j \in G_1$.\\
  16 \>\> Extend $Q_j$ to $Q$ and $A_j$ to $A$, adding the missing
  faces in $G^*_{1,s_jt_k}$ that are incident to $a_\infty^*$.\\
  17\>\> {\sc RecursiveFlow}$(G_1, \{s_1, s_2, \ldots, s_j\}$, $\{s_{\ell+1}, t_{\ell+1}, \ldots, s_m, t_m\}, \{(Q_1,A_1), \ldots, (Q_j,A_j)\})$
 \end{tabbing}}

\paragraph{Running time and correctness of {\sc RecursiveFlow}}

By Observation~\ref{obs:res-cap}, Step 4 results in the priorities
reflecting exactly the residual capacities of the darts in $P$ after
saturating $\fvec$.  $G_1$ and $G_2$ are the same as the subgraphs
created in {\sc AbstractRecursiveFlow}, as are the subproblems
considered.  The removed darts create a new boundary and so maintain
triangulation of the finite faces.  Step 12 can be done in
$O(\log n)$ per new face added (Appendices~\ref{app:pq}
and~\ref{app:order}).  Adding the faces can be achieved by a left-first search
from $Q$ (or from $Q_j$ to $Q$); this creates the queue and order along the boundary
of the graph.  In order to combine the orders $A_j$ and $A$ in line 16, we observe that the order $A_j$ is guaranteed to be right of the order $A$ when they are joined together.  The {\sc DSOrder} data structure allows us to concatenate these orders efficiently (details in Appendix~\ref{app:order}).

Finally, we argue that the entire algorithm requires only $O(n)$
queries to priority queue and {\sc DSOrder} data structure.
The biased-search algorithm uses $O(k)$ priority-queue and {\sc
  DSOrder} queries where $k$ is the size of the search tree discovered
(Section~\ref{sec:biased}).  This is in part due to the triangulation
of the finite faces; the degree of the vertices from which we search
during the biased-search algorithm have degree 3, so the 0-length
darts leaving a vertex can be determined in constant time.

For the subproblem $G_1$, we start with queues that have already been
initialized, so, as argued at the end of Section~\ref{sec:biased}, we
essentially pick up the search where we left off, not repeating any
computation at the boundary where we left off (the rightmost path in a
previous flow).  For the subproblem $G_2$, $P$ forms part of the
boundary and so part of the queue/order ending at $t_\ell$ appear in
this subgraph.  However, by Lemma~\ref{lem:path-bdy}, no residual path
intersects $P$.  Since the finite faces are triangulated, no path can
intersect a face adjacent to $P$ without intersecting $P$.  Therefore,
none of the faces in the queue/order along $P$ will be used in the
subproblem corresponding to $G_2$.  It follows that there are a
constant number of data-structure queries per finite face of the
original graph.
{\let\thefootnote\relax\footnotetext{{\em Acknowledgements:} This material is based upon work supported by
the National Science Foundation under Grant No.\ CCF-0963921.}}

\appendix

\section{Priority queues with offsets} \label{app:pq}

We show how to efficiently change all the priorities in a queue by a
fixed amount.  This will be used when we wish to merge two priority
queues whose relative priorities differ by a constant.  That is, we
have two priority queues $P$ and $Q$ that we want to merge, but the
priorities of the items in $P$ are {\em offset} from those in $Q$ by
some amount $o$.  We illustrate this for a
binomial-heap implementation of priority queues, but this technique is
not limited to a specific implementation (although the details of
handling the offsets will depend on the implementation).

For the purposes of this discussion the details of a binomial heap,
beyond the fact that it is a set of rooted trees, are irrelevant.  We
refer the reader to any data structures textbook for details.  We will argue
that the standard operations (insert, find minimum, delete minimum,
decrease key and merge) will have the same asymptotic running time
with offsets as without.  To do so, we annotate the edges of the trees
in the heap with weights, initially zero.  We give the roots of the
trees a dummy parent edge so that every item in the queue (node in a
tree) $x$ has a unique parental edge weight $w(x)$.  We say that node
$x$ has a local priority $p_\ell(x)$ and a global priority $p(x)$
where $p(x)$ is the sum of $p_\ell(x)$ plus the parental edge weights
on the path to the root of the binomial tree containing $x$ (including
the weight of the dummy root edge).
Initially the global priorities are the same as the local priorities.
We will maintain that the heap property holds for the
global priorities (ie.\ my children's global priorities are lower than
mine).

We describe the modifications we make to the binomial-heap-based
priority queue operations:
\begin{description} 
  \item[insert] Unchanged as insert reduces to merge.
  \item[find min] The minimum priority element is guaranteed to be a
    root of one of the trees.  When comparing the roots of the trees,
    first sum the local priority and dummy root edge weight.
  \item[delete min] The standard operation is to delete the root that
    is the minimum priority element and then merge the resulting child
    trees with the remaining trees.  We first add the weight of the
    dummy root edge to the weights of the child edges; these child
    edges become dummy parent edges of the trees before they are
    merged.
  \item[decrease key] The standard operation traverses the path from
    the node in question, $x$ to the root and swaps nodes that violate
    the heap property.  First compute the global priorities of the
    nodes on the $x$ to root path. Then traverse to the to-root path:
    say $x$ is a child of $y$ such that $p(x) < p(y)$; let $w$ be the
    weight of the edge $xy$.  Swap $x$ and $y$, add $w$ to
    $p_\ell(x)$ and subtract $w$ from $p_\ell(y)$.
  \item[merge] If we want to merge heap $P$ with heap $Q$ in such a way that
    the priorities in $P$ are by an offset $o$ higher than those in $Q$, $o$ is added to the weight of the dummy root edges of $P$
    and in comparing the priorities of the roots of trees in $P$ to
    those in $Q$, the global priorities are used.  Merging binomial
    heaps is otherwise trivial.
\end{description}
We note that our modifications to not increase the asymptotic
complexity of the operations.  Although we do not need to maintain
local priorities for our algorithm, we point out that local priorities
can be retained.  However, in the decrease-key operation, the weight
of sibling edges would need to be modified as well, and, for binomial
heaps, would require $O(\log^2 n)$ time.

\section{Maintaining order} \label{app:order}

In order to maintain the left-to-right order of faces in the priority
queue we refer to an order maintenance data structure {\sc DSOrder} due to Dietz and
Sleator~\cite{DS87}. {\sc DSOrder} supports the following operations:

\begin{enumerate}
\item $Insert(X; Y)$: Insert a new element $Y$ immediately after
  element $X$ in the total order.
\item $Delete(X)$: Remove an element $X$ from the total order.
\item $Order(X; Y)$: Determine whether $X$ precedes $Y$ in the total
  order
\end{enumerate}

While there are other data structures that are
more efficient asymptotically~\cite{BCDFZ02}, {\sc DSOrder} is attractive for its
simplicity, as it only relies on basic data structures.  {\sc DSOrder}
is implemented as a circularly linked list that implicitly
encodes the label bits to represent paths in a hypothetical $2-4$ tree
and uses $2$'s complement arithmetic and a wrapping modulo to
efficiently perform renumbering, giving:

\begin{theorem}\cite{DS87}
  The amortized time to do {\em Insert} on a list containing $n$
  records  is $O(\log n)$, and the amortized (and  worst-case) time to do {\em Delete} or {\em Order} is $O(1)$. 
\end{theorem}

{\sc DSOrder} generally draws its labels from integers in $\{0,\ldots, M - 1\}$,
where $M$ is sufficiently large\footnote{$M > n^2$, where $n$ is the
  size of the order.}. Since in our algorithm every face in a newly created order is {\em right of} the faces in the
previous order, we modify this range as we move left-to-right to make
simple concatenation possible. I.e. if $n_i$ is the largest label in
the order $A_i$, the labels for $A_{i+1}$ are drawn from
$\{n_i+1,\ldots, n_i+ M\}$, where $M$ is large \wrt the size of the
graph. Then, an order $B$ created after an order $A$, can be
appended to $A$ in constant time via standard linked list operations.


\begin{thebibliography}{13}

\bibitem{BCDFZ02}
M.~Bender, R.~Cole, E.~Demaine, M.~Farach-Colton, J.~Zito.
\newblock Two Simplified Algorithms for Maintaining Order in a List.
\newblock In {\em Proc. ESA}, pages 152--164, 2002.

\bibitem{BH13}
G.~Borradaile and A.~Harutyunyan.
\newblock Maximum st-flow in directed planar graphs via shortest paths.
\newblock To appear in {\em Proc. IWOCA}, 2013.

\bibitem{BK09}
G.~Borradaile and P.~Klein.
\newblock An ${O}(n \log n)$ algorithm for maximum st-flow in a directed planar
  graph.
\newblock {\em J. of the ACM}, 56(2):1--30, 2009.

\bibitem{BKMNW11}
G.~Borradaile, P.~Klein, S.~Mozes, Y.~Nussbaum, C.~Wulff-Nilsen.
\newblock Multiple-Source Multiple-Sink Maximum Flow in Directed
Planar Graphs in Near-Linear Time.
\newblock In {\em Proc. FOCS}, pages 170--179, 2011.

\bibitem{Borradaile-thesis}
G.~Borradaile.
\newblock {\em Exploiting Planarity for Network Flow and Connectivity
  Problems}.
\newblock PhD thesis, Brown University, 2008.

\bibitem{DS87}
P.~Dietz and D.~Sleator.
\newblock Two algorithms for maintaining order in a list.
\newblock In {\em Proc. STOC}, pages 365--372, 1987.

\bibitem{FF56}
C.~Ford and D.~Fulkerson.
\newblock Maximal flow through a network.
\newblock {\em Canadian J. Math.}, 8:399--404, 1956.

\bibitem{GH05}
A.~Goldberg and C.~Harrelson.
\newblock Computing the shortest path: A search meets graph theory.
\newblock In {\em Proc. SODA}, pages 156--165, 2005.

\bibitem{Hassin81}
R.~Hassin.
\newblock Maximum flow in $(s,t)$ planar networks.
\newblock {\em IPL}, 13:107, 1981.

\bibitem{HKRS97}
M.~Henzinger, P.~Klein, S.~Rao, S. Subramanian.
\newblock Faster shortest-path algorithms for planar graphs.
\newblock {\em JCSS}, 55(1):3--23, 1997.

\bibitem{KNK93}
S.~Khuller, J.~Naor, P.~Klein.
\newblock The lattice structure of flow in planar graphs.
\newblock {\em SIAM J. on Disc. Math.}, 6(3):477--490, 1993.

\bibitem{MN95}
G.~Miller and J.~Naor
\newblock Flow in planar graphs with multiple sources and sinks.
\newblock {\em SIAM J. on Comp.}, 24(5):1002--1017, 1995.

\bibitem{Whitney1933}
H.~Whitney
\newblock Planar Graphs.
\newblock {\em Fundamenta Mathematicae}, 21:73--84, 1933.

\end{thebibliography}
\end{document}